%% file: arxiv.tex
\documentclass{article}
\usepackage{spconf,amsmath,graphicx}
\input{macros}

\title{Spectral radii of asymptotic mappings and the convergence speed of the standard fixed point algorithm}
\name{R.~L.~G.~Cavalcante and S.~Sta\'nczak\thanks{\scriptsize This research was supported by Grant STA 864/9-1 from the German Research Foundation (DFG). This work has also been performed in the framework of the Horizon 2020 project ONE5G (ICT - 760809) receiving funds from the European Union. The authors would like to acknowledge the contributions of their colleagues in the project, although the views expressed in this contribution are those of the authors and do not necessarily represent the project. }}
\address{Fraunhofer Heinrich Hertz Institute and Technical University of Berlin}

\begin{document}
\ninept

\maketitle
\begin{abstract}
	Important problems in wireless networks can often be solved by computing fixed points of standard or contractive interference mappings, and the conventional fixed point algorithm is widely used for this purpose.
Knowing that the mapping used in the algorithm is not only standard but also contractive (or only contractive) is valuable information because we obtain a guarantee of geometric convergence rate, and the rate is related to a property of the mapping called modulus of contraction. To date, contractive mappings and their moduli of contraction have been identified with case-by-case approaches that can be difficult to generalize. To address this limitation of existing approaches, we show in this study that the spectral radii of asymptotic mappings can be used to identify an important subclass of contractive mappings and also to estimate their moduli of contraction. In addition, if the fixed point algorithm is applied to compute fixed points of positive concave mappings, we show that the spectral radii of asymptotic mappings provide us with simple lower bounds for the estimation error of the iterates. An immediate application of this result proves that a known algorithm for load estimation in wireless networks becomes slower with increasing traffic.

\end{abstract}
\section{Introduction}
\label{sec:intro}
The objective of this study is to investigate convergence properties of the sequence $(\signal{x}_n)_{n\in\Natural}$ generated by the following instance of the \emph{standard fixed point algorithm}:
\begin{align}
\label{eq.fpi}
\signal{x}_{n+1}=T(\signal{x}_n),
\end{align}
where $\signal{x}_{1}\in\real_{+}^N$ is an arbitrary initial point; $\real_+^N$ denotes the set of nonnegative vectors of dimension $N$; and $T:\real^N_+\to\real^N_{+}$ is a standard interference mapping as defined in \cite{yates95} or a ($c$-)contractive mapping as defined in \cite{fey2012}, or both. Previous studies \cite{yates95,fey2012} have shown that, if $T$ is a standard interference mapping with $\mathrm{Fix}(T):=\{\signal{x}\in\real_+^N~|~\signal{x}=T(\signal{x})\}\neq\emptyset$ or a contractive mapping, then $\mathrm{Fix}(T)$ is a singleton, and the sequence generated by \refeq{eq.fpi} converges to the fixed point $\signal{x}^\star\in\mathrm{Fix}(T)$. The algorithm in \refeq{eq.fpi} plays a pivotal role in many power and resource allocation mechanisms in wireless networks \cite{yates95,renato14SPM,ho2015,martin11,huang1998rate,renato17load,renato2016,renato2016maxmin,slawomir09,feh2013,nuzman07,fey2012,boche2008,martin2006robust}, so establishing its convergence rate is a problem of significant practical importance \cite{martin11,fey2012,huang1998rate,renato2016}.

If the mapping $T$ in \refeq{eq.fpi} is only a standard interference mapping, then the fixed point algorithm  can be particularly slow because we can have sublinear convergence rate \cite[Example 1]{fey2012}. This fact has motivated the authors of \cite{fey2012} to introduce the above-mentioned $c$-contractive interference mappings, where $c\in~[0,1[$ is an intrinsic property of the mapping called {\it modulus of contraction} (see Definition~\ref{definition.mappings} in Sect.~\ref{sect.pre} for details). In particular, by knowing that the mapping $T$ in \refeq{eq.fpi} is contractive, we rule out the possibility of sublinear convergence rate. More precisely, by using \refeq{eq.fpi} with a $c$-contractive mapping $T$ to estimate $\signal{x}^\star\in\mathrm{Fix}({T})$, the error $\|\signal{x}_n-\signal{x}^\star\|$ of the estimate $\signal{x}_n$ at iteration $n\in\Natural$ is upper bounded by \cite{fey2012}
\begin{align}
\label{eq.conv_rate}
\|\signal{x}_n-\signal{x}^\star\|\le c^{n-1} B \|\signal{x}_1-\signal{x}^\star\|,
\end{align}
 where $B\in\real_+$ is a parameter that depends on the choice of the norm $\|\cdot\|$. Therefore, with knowledge of the smallest modulus of contraction $c$, we can evaluate whether the recursion in \refeq{eq.fpi} can obtain a good estimate of $\signal{x}^\star\in\mathrm{Fix}(T)$ with few iterations. However, simple and general approaches to verify whether a  mapping is contractive have not been proposed in \cite{fey2012}. Furthermore, that study has not considered computationally efficient methods to obtain the smallest modulus of contraction. 

Against this background, in this study we show that information about the smallest modulus of contraction of a convex contractive interference mapping can be obtained from the spectral radius of its associated asymptotic mapping, a concept recently introduced in  \cite{renato17globalsip,renato2017performance}. We further show easily verifiable sufficient conditions to determine whether a given mapping is contractive. In addition, we give lower bounds for the estimation error of the iterates in \refeq{eq.fpi} with positive concave mappings that are not necessarily contractive. Unlike the bounds in previous studies \cite[Ch.~5]{martin11} \cite{huang1998rate,fey2012}, those derived here only depend on parameters that are easy to compute in practice, and we do not assume that the mappings used in \refeq{eq.fpi} are constructed by combining a finite number of affine mappings. As an application of the results in this study, we show bounds for the estimation error of the iterates generated by \refeq{eq.fpi} with a \emph{nonlinear} mapping widely used to estimate the load of base stations in wireless networks \cite{Majewski2010,siomina12,renato14SPM,renato2016,siomina2014,ho2014data}. In this application, our bounds give a formal proof  that the algorithm for load estimation becomes slower with increasing traffic.

\section{Preliminaries}
\label{sect.pre}
In this section we establish notation and review  the main mathematical concepts used in this study. In more detail, the sets of nonnegative and positive reals are denoted by, respectively, $\real_+$ and $\real_{++}$. Inequalities involving vectors should be understood coordinate-wise. A norm $\|\cdot\|$ in $\real^N$ is \emph{monotone} if $(\forall \signal{x}\in\real_+^N)(\forall \signal{y}\in\real_+^N)~ \signal{x}\le\signal{y}\Rightarrow \|\signal{x}\|\le\|\signal{y}\|$. We say that a sequence $(\signal{x}_n)_{n\in\Natural}\subset\real_+^N$ converges  to $\signal{x}^\star$ if $\lim_{n\to\infty}\|\signal{x}_n-\signal{x}^\star\|=0$ for some (and hence for every) norm $\|\cdot\|$ in $\real^N$, and in this case we also write $\signal{x}_n\to\signal{x}^\star$. Given a norm $\|\cdot\|$ and a sequence $(\signal{x}_n)_{n\in\Natural}\subset\real_+^N$, if $(\exists c\in~[0,1[)(\exists B\in\real_+)(\forall n\in\Natural)~ {\|\signal{x}_{n+1}-\signal{x}^\star\|}\le c^n B \|\signal{x}_1-\signal{x}^\star\|$, then we say that $(\signal{x}_n)_{n\in\Natural}$ (or the algorithm generating the sequence) converges \emph{geometrically fast}. 
The \emph{(effective) domain} of a function $f:\real^N\to\real\cup\{-\infty,\infty\}$ is the set given by $\mathrm{dom}f:=\{\signal{x}\in\real^N~|~f(\signal{x})<\infty\}$, and $f$ is {\it proper} if $\mathrm{dom}f\neq\emptyset$ and $(\forall \signal{x}\in\real^N)~f(\signal{x})>-\infty$.

\begin{definition}
	\label{definition.mappings} (Standard and $c$-contractive interference mappings:) 	Consider the following statements for a \emph{continuous} mapping $T:\real^N_+\to\real_{++}^N$:
	\begin{itemize}
		\item[(i)] [monotonicity] $(\forall \signal{x}\in\real^N_{+})(\forall \signal{y}\in\real^N_{+}) ~ \signal{x}\ge\signal{y} \Rightarrow T(\signal{x})\ge T(\signal{y})$
		\item[(ii)] [scalability] $(\forall \signal{x}\in\real^N_+)$ $(\forall \alpha>1)$  $\alpha {T}(\signal{x})>T(\alpha\signal{x})$.
		\item[(iii)] [contractivity] $(\exists (\signal{v}, c) \in\real_{++}^N\times [0,1[) (\forall \signal{x}\in\real^N_+) (\forall \epsilon>0) ~ T(\signal{x}+\epsilon\signal{v})\le T(\signal{x})+c\epsilon\signal{v}$
	\end{itemize}
	If (i) and (ii) are satisfied, then $T$ is said to be a \emph{standard interference mapping} \cite{yates95}. If $T$ satisfies (i) and (iii), then $T$ is called a \emph{contractive interference mapping}. In this case, if a scalar $c$ with the property in (iii) is known, then $c$ is called a \emph{modulus of contraction} for $T$, and we also say that $T$ is $c$-\emph{contractive} to emphasize this knowledge \cite{fey2012}.
\end{definition}

A mapping $T:\real_+^N\to\real_+^N$ is said to be \emph{concave} (respectively \emph{convex}) if each coordinate function is concave (respectively convex). Recall from the Introduction that the set of fixed points of a mapping $T:\real_+^N\to\real_+^N$ is denoted by $\mathrm{Fix}(T):=\{\signal{x}\in\real_+^N~|~T(\signal{x})=\signal{x}\}$. If $T$ is contractive, then $\mathrm{Fix}(T)$ is a singleton \cite{fey2012}. If $T$ is a standard interference mapping, then $\mathrm{Fix}(T)$ is either a singleton or the empty set \cite{yates95}.

Given a proper function $f:\real^N\to\real\cup\{\infty\}$, we say that $f_\infty:\real^N\to\real\cup\{-\infty\}\cup\{ \infty\}:\signal{x}\mapsto\liminf_{t\to\infty,\signal{y}\to\signal{x}}{f(t\signal{y})}/{t}$ is the \emph{asymptotic function} associated with $f$ \cite[Ch.~2.5]{aus03}, and note that $f_\infty$ is positively homogeneous [i.e., $(\forall \signal{x}\in\real^N)(\forall \alpha> 0)f(\alpha\signal{x})=\alpha~f(\signal{x})$] and lower semicontinuous \cite[Proposition~2.5.1]{aus03}. Asymptotic functions associated with convex functions have the following useful property:

\begin{fact}
	\label{fact.sup_conv}
	\cite[Proposition~2.5.2]{aus03} Let $f:\real^N\to\real\cup\{\infty\}$ be proper, lower semicontinuous, and convex. Then 
	\begin{align*}
	(\forall\signal{d}\in\real^N)~f_\infty(\signal{d}) = \sup\{f(\signal{x}+\signal{d})-f(\signal{x})~|~\signal{x}\in\mathrm{dom}f\}. 
	\end{align*}
\end{fact}

In the next Lemma, we show a result related to Fact~\ref{fact.sup_conv} for nonnegative concave functions. We omit the proof because of the space limitation.

\begin{lemma}
	\label{lemma.inf_conc}
	Let $f:\real^N\to\real_+\cup\{\infty\}$ be a function such that  $\mathrm{dom}~f=\real_+^N$. Assume that $f$ is continuous and concave if restricted to its domain. Then $(\forall\signal{d}\in \mathrm{dom}~f)$ $f_\infty(\signal{d}) = \inf\{f(\signal{x}+\signal{d})-f(\signal{x})~|~\signal{x}\in\mathrm{dom}~f\}$.  
\end{lemma}

We now introduce a slight generalization of the concept of asymptotic mappings given in \cite{renato17globalsip,renato2017performance}.
\begin{definition}
	\label{def.amap}
	(Asymptotic mappings:) Let $T:\real_+^N\to\real_+^N:\signal{x}\mapsto [T^{(1)}(\signal{x}),\cdots,T^{(N)}(\signal{x})]$ be a mapping such that, for each $i\in\{1,\ldots,N\}$, the function $T^{(i)}:\real^N\to\real_+\cup\{\infty\}$ is proper and  $\mathrm{dom}~T^{(i)}=\mathrm{dom}~T^{(i)}_\infty=\real_+^N$. For these mappings, we say that $T_\infty:\real_+^N\to\real_+^N:\signal{x}\mapsto[T^{(1)}_\infty(\signal{x}),\cdots,T^{(N)}_\infty(\signal{x})]$ is the \emph{asymptotic mapping} associated with $T$.
\end{definition}
If $T:\real^N_+\to\real^N_+$ is a continuous concave mapping, a standard interference mapping, or a convex mapping having an asymptotic mapping, then we can use the following analytical simplification to obtain the asymptotic mapping \cite{renato17globalsip,renato2017performance}\cite[Corollary~2.5.3]{aus03}: $(\forall\signal{x}\in\real_+^N)~T_\infty(\signal{x}) = \lim_{t\to\infty} (1/t)T(t\signal{x})$. The spectral radius $\rho(T_\infty)$ of a \emph{continuous} and \emph{monotonic} (see Definition~\ref{definition.mappings}(i)) asymptotic mapping $T_\infty$ is the value given by $\rho(T_\infty) := \sup\{\lambda\in\real_+~|~(\exists \signal{x}\in\real_+^N\backslash\{\signal{0}\})~ T_\infty(\signal{x})=\lambda\signal{x}\} \in \real_+$, and we recall that there always exists an {\it eigenvector} $\signal{x}\in\real_+^N$ satisfying $\rho(T_\infty)\signal{x}=T_\infty(\signal{x})$ \cite{nussbaum1986convexity}. The next fact is crucial to prove our main contributions.

\begin{fact}
	\label{fact.existence}
	\cite{renato17globalsip} Let $T:\real^N_+\to\real_{++}^N$ be a standard interference mapping. Then $\mathrm{Fix}(T)\neq\emptyset$ if and only if $\rho(T_\infty)<1$.
\end{fact}

\section{Convergence properties of the standard fixed point algorithm}

By \refeq{eq.conv_rate}, the sequence generated by \refeq{eq.fpi} with a $c$-contractive mapping $T$ has the desirable property of converging geometrically fast, and the convergence speed is directly related to the modulus of contraction $c$. Therefore, identifying contractive mappings and estimating their moduli of contraction are important tasks. In Sect.~\ref{sect.convex}, we prove that the spectral radii of asymptotic mappings can be used for these tasks if $T$ is convex. Then, in Sect.~\ref{sect.concave} we show that, if the fixed point algorithm in \refeq{eq.fpi} is used with an arbitrary (continuous) positive concave mapping $T$, then the spectral radius of $T_\infty$ provides us with information about the fastest convergence speed we can expect from the algorithm. All these results are especially useful if we can easily evaluate the spectral radii of arbitrary asymptotic mappings, so we start by showing in Sect.~\ref{sect.spec_radius} simple algorithms for this purpose. These algorithms also enable us to obtain information about an eigenvector associated with the spectral radius.

\subsection{Spectral radius of asymptotic mappings}
\label{sect.spec_radius}

Let $T_\infty:\real^N_+\to\real^N_+$ be a continuous asymptotic mapping associated with a continuous mapping $T:\real^N_+\to\real^N_+$ satisfying property (i) in Definition~\ref{definition.mappings}. It can be verified that $T_\infty$ also satisfies property (i). If $T_\infty$ is in addition concave and primitive, in the sense that\footnote{$T^m_\infty$ denotes the $m$-fold composition of $T_\infty$ with itself.} $(\forall \signal{x}\in\real^N_+\backslash\{\signal{0}\})(\exists p\in\Natural)(\forall m\ge p) T^m_\infty(\signal{x})>\signal{0},$
then the sequence $(\signal{x}_n)_{n\in\Natural}$ generated by 
\begin{align}
\label{eq.krause_iter}
\signal{x}_{n+1} = \dfrac{1}{\|T_\infty(\signal{x}_n)\|}T_\infty(\signal{x}_n), ~\signal{x}_1\in\real_{+}^N\backslash\{\signal{0}\}
\end{align}
with an arbitrary monotone norm $\|\cdot\|$ 
converges to a point $\signal{x}^\star\in\real_{++}^N$ such that $T_\infty(\signal{x}^\star)=\|T_\infty(\signal{x}^\star)\| ~ \signal{x}^\star$ and $\|\signal{x}^\star\| = 1$ \cite{krause1986perron,krause01}. Therefore, by \cite[Lemma~3.3]{nussbaum1986convexity}, we conclude that  $\rho(T_\infty)=\|T_\infty(\signal{x}^\star)\|$. In practical terms, the iteration in \refeq{eq.krause_iter} is a simple algorithm to compute the spectral radius and a corresponding eigenvector of an asymptotic mapping, provided that the assumptions mentioned above are valid. In more challenging cases in which existing results such as those in \cite{krause1986perron,krause01} does not necessarily guarantee convergence of \refeq{eq.krause_iter} to a point $\signal{x}^\star$ satisfying $\rho(T_\infty)=\|T_\infty(\signal{x}^\star)\|$, we propose an approach based on the following result (the proof is omitted because of the space limitation):
\begin{proposition}
	\label{prop.alg_spec}
	 Let $T_\infty:\real^N_{+}\to\real_{+}^N$ be a continuous asymptotic mapping satisfying the monotonicity property in Definition~\ref{definition.mappings}, and consider the mapping  $T_\epsilon:\real^N_{+}\to\real_{+}^N:\signal{x}\mapsto T_\infty(\signal{x})+\epsilon\signal{1}$, where $\signal{1}\in\real^N$ denotes the vector of ones and $\epsilon>0$ is arbitrary. For a given parameter $p>0$, let the sequence $(\signal{x}_{p,n})_{n\in\Natural}$ be generated by 
\begin{align}
\label{eq.epsilon_iter}
\signal{x}_{p,n} := \dfrac{p}{\|T_\epsilon(\signal{x}_{p,n})\|} T_\epsilon(\signal{x}_{p,n}),
\end{align}	
where $\signal{x}_{p,1}\in\real_{+}^N$ is arbitrary, and $\|\cdot\|$ is a monotone norm. Then we have the following:
\begin{itemize}
\item[(i)] $T_\epsilon$ is a standard interference mapping.
\item[(ii)] For every $p>0$, the sequence $(\signal{x}_{p,n})_{n\in\Natural}$ converges to a point $\signal{x}_p^\star\in\real_{++}^N$ satisfying $T_\epsilon(\signal{x}_p^\star)=(\|T_\epsilon(\signal{x}_p^\star)\|/p) ~ \signal{x}_p^\star$ and $\|\signal{x}_p^\star\| = p$.
\item[(iii)] $(\forall p>0)~ \rho(T_\infty)\le \|T_\epsilon(\signal{x}_p^\star)\|/p $ 
\item[(iv)] $\lim_{p\to \infty} \|T_\epsilon(\signal{x}_p^\star)\|/p = \rho(T_\infty)$
\item[(v)] If $(p_n)\subset\real_{++}$ is a  sequence satisfying $\lim_{n\to\infty}p_n=\infty$, then any accumulation point of $(\signal{x}^\star_{p_n})_{n\in\Natural}$ is an eigenvector of $T_\infty$ associated with the eigenvalue $\rho(T_\infty)$.
\end{itemize}

\end{proposition}

In simple terms, Proposition~\ref{prop.alg_spec}(iii)-(iv) shows that the spectral radius $\rho(T_\infty)$ of any asymptotic mapping $T_\infty$ that is monotonic and continuous can be estimated with any arbitrary precision by using \refeq{eq.epsilon_iter}. Informally, given an arbitrary scalar $\epsilon>0$, if $n\in\Natural$ and ${p}\in\real_{++}$ are sufficiently large, then $\rho(T_\infty)\approx \|T_\epsilon(\signal{x}_{p,n})\|/p$, where $(\signal{x}_{p,n})_{n\in\Natural}$ is the sequence generated by \refeq{eq.epsilon_iter}. Furthermore, by assuming that  $(\signal{x}^\star_{p_n})_{n\in\Natural}$ in Proposition~\ref{prop.alg_spec}(v) converges, then $\signal{x}_{p,n}$ with the above parameters is an approximation of an eigenvector of $T_\infty$ associated with the spectral radius $\rho(T_\infty)$.

\subsection{Convex mappings}
\label{sect.convex}
Checking whether a continuous and monotonic mapping $T:\real_+^N\to\real_{++}^N$ is $c$-contractive may be challenging because proving the existence of a tuple $(\signal{v},~c)\in\real_{++}^N\times~[0,1[$ with the property in Definition~\ref{definition.mappings}(iii) may be difficult. However, as we show in the next proposition, if $T$ is convex (as common in many robust wireless resource allocation problems \cite{boche2008,martin2006robust}), then knowledge of the spectral radius of $T_\infty$, assuming that $T_\infty$ exists, can be used to determine whether $T$ is contractive. 

\begin{proposition}
\label{prop.modu_contract}
	Let $T:\real_+^N\to\real_{++}^N$ be a continuous convex mapping that has an associated continuous asymptotic mapping $T_\infty:\real_+^N\to\real_{+}^N$. Further assume the following: (i) there exists a (strictly) positive vector $\signal{v}\in\real_{++}^N$ such that  $T_\infty(\signal{v})=\rho(T_\infty)\signal{v}$, (ii) $T$ satisfies the monotonicity property in Definition~\ref{definition.mappings}, and (iii) $\rho(T_\infty)<1$. Then $T$ is $c$-contractive for any $c \in [\rho(T_\infty),1[~\subset~[0,1[$.
\end{proposition}

\begin{proof}
	We only have to show that the property in Definition~\ref{definition.mappings}(iii) can be satisfied with $c=\rho(T_\infty)\ge0$. To this end, let $\signal{v}\in\real^N_{++}$ be a vector with the property in assumption (i).  Now, by Fact~\ref{fact.sup_conv} and positive homogeneity of asymptotic mappings, we deduce  
	$(\forall \epsilon>0)(\forall\signal{x}\in\real^N_+) ~ T(\signal{x}+\epsilon \signal{v}) \le T(\signal{x})+ T_\infty(\epsilon\signal{v})=T(\signal{x})+\rho(T_\infty)\epsilon\signal{v}$, and the claim follows. 
\end{proof}

We note that there are many simple results to  verify assumption (i) in Proposition~\ref{prop.modu_contract} without explicitly computing a so-called (nonlinear) eigenvector $\signal{v}\in\real_{++}$ \cite{gau04}. In addition, neither assumption (ii) nor assumption (iii) can be dropped. The former is required because of the definition of contractive mappings, and the latter is also necessary because, as shown below, the spectral radius of the asymptotic mapping $T_\infty$ associated with a $c$-contractive mapping $T$ is a lower bound for the modulus of contraction $c\in [0,1[$.

\begin{proposition}
	\label{prop.inf_cont}
	Let $T:\real_+^N\to\real_{++}^N$ be $c$-contractive and convex. Then $T$ has a continuous asymptotic mapping $T_\infty:\real_{+}^N\to\real_+^N$ satisfying $\rho(T_\infty)\le c\in~[0,1[$.
\end{proposition}
\begin{proof}
	 By definition, if $T$ is $c$-contractive, there exists $(\signal{v},c)\in\real_{++}^N\times [0,1[$ such that
	\begin{align}
	\label{eq.mod_c}
	(\forall \epsilon>0)(\forall \signal{x}\in\real_+^N) T(\signal{x}+\epsilon\signal{v}) - T(\signal{x}) \le c\epsilon \signal{v}. 
	\end{align}
	Denote by $T^{(i)}:\real_+^N\to\real_{++}^N$ the $i$th coordinate function of the mapping $T$; i.e., $(\forall\signal{x}\in\real_+^N) [T^{(1)}(\signal{x}),\cdots,T^{(N)}(\signal{x})]:=T(\signal{x})$. Since $T$ is continuous and convex, by Fact~\ref{fact.sup_conv} we have
	\begin{multline}
	\label{eq.coord_sup}
	(\forall \signal{d}\in\real_+^N)(\forall i\in\{1,\cdots,N\})\\~T^{(i)}_\infty(\signal{d}) := \sup\{T^{(i)}(\signal{x}+\signal{d})-T^{(i)}(\signal{x})~|~\signal{x}\in\real_+^N\}.
	\end{multline}
	To prove that $T_\infty(\signal{x}):= [T^{(1)}_\infty(\signal{x}),\cdots,T^{(N)}_\infty(\signal{x})]\ge\signal{0}$ for $\signal{x}\in\real_+^N$ is the asymptotic \emph{mapping} associated with $T$ (in the sense of Definition~\ref{def.amap}), we need to show that $\mathrm{dom}~T^{(i)}_\infty=\real_+^N$ for each $i\in\{1,\ldots, N\}$
	. To this end, take the coordinatewise supremum in \refeq{eq.mod_c} over $\signal{x}\in\real_+^N$ and apply \refeq{eq.coord_sup} with $\signal{d}=\epsilon \signal{v}$ and an arbitrary $\epsilon>0$ to obtain 
	\begin{align}
	\label{eq.local_ineq}
	 [T^{(1)}_\infty(\epsilon \signal{v}),\cdots,T^{(N)}_\infty(\epsilon\signal{v})]=:T_\infty(\epsilon\signal{v})\le c \epsilon\signal{v}.
	\end{align}

	By positivity of $\signal{v}$, for an arbitrary $\signal{x}\in\real_+^N$, there exists $\eta>0$ such that $\signal{x}\le \eta\signal{v}$. Since $\epsilon>0$ in \refeq{eq.local_ineq} can be chosen arbitrarily, we can use $\eta = \epsilon$ and monotonicity of $T$ to deduce $T_\infty(\signal{x}) \le T_\infty(\eta \signal{v})\le c \eta \signal{v}\in\real_+^N$. As a result, we have $T^{(i)}_\infty(\signal{x})<\infty$ for all $\signal{x}\in\real_+^N$ and all $i\in\{1,\ldots,N\}$ as claimed. {(We can also show that $T_\infty$ is continuous in $\real_+^N$, but we omit the proof because of the space limitation.)} With the inequality in \refeq{eq.local_ineq} and continuity of $T_\infty$, we also obtain $\rho(T_\infty) \le c$ by  \cite[Lemma~3.3]{nussbaum1986convexity}, and the proof is complete.
\end{proof}

We now show a useful relation between contractive and standard interference mappings. From a practical perspective, the next result and \refeq{eq.conv_rate} reveal that many existing iterative algorithms for power control in wireless networks converge geometrically fast. Furthermore, the inequality in \refeq{eq.conv_rate}, Proposition~\ref{prop.modu_contract}, and Proposition~\ref{prop.inf_cont} show that the concept of spectral radius of asymptotic mappings provides us with information about the convergence speed of these algorithms.
\begin{proposition}
	\label{prop.convex_sim}
	Let $T:\real_+^N\to\real_{++}^N$ be a convex standard interference mapping. Then $\rho(T_\infty)<1$ is a sufficient and necessary condition for $T$ to be contractive.
\end{proposition}

\begin{proof}
	By Fact~\ref{fact.existence}, if $\rho(T_\infty)\ge 1$, then $\mathrm{Fix}(T)=\emptyset$, so $T$ cannot be contractive because contractive mappings have a fixed point \cite{fey2012}. Therefore, $\rho(T_\infty)<1$ is a necessary condition. To prove sufficiency, we only need to show that property (iii) in Definition~\ref{definition.mappings} is satisfied if $\rho(T_\infty)<1$. By Fact~\ref{fact.existence}, if $\rho(T_\infty)<1$ then there exists $\signal{x}^\star\in\real_{++}^N$ such that $\signal{x}^\star = T(\signal{x}^\star)$.  By \cite[Lemma~1(ii)]{renato2017performance} and $T_\infty(\signal{x}^\star)=\lim_{t\to\infty}T(t\signal{x}^\star)/t$ \cite[Corollary~2.5.3]{aus03}, we have $T_\infty(\signal{x}^\star)<T(\signal{x}^\star)=\signal{x}^\star$. As a result, there exists $c\in[0,1[$ such that $\epsilon T_\infty(\signal{x}^\star)\le \epsilon c \signal{x}^\star$ for all $\epsilon>0$. Therefore, by the positive homogeneity of asymptotic functions, we have $T_\infty(\epsilon \signal{x}^\star)\le \epsilon c \signal{x}^\star$ for all $\epsilon>0$. By Fact~\ref{fact.sup_conv}, we conclude that $T(\signal{x}+\epsilon\signal{x}^\star)-T(\signal{x})\le T_\infty(\epsilon\signal{x}^\star) \le \epsilon c \signal{x}^\star$ for every $\signal{x}\in\real_{+}^N$, and the desired result follows.
\end{proof}

\subsection{Concave mappings}
\label{sect.concave}

We now proceed to study convergence properties of the algorithm in \refeq{eq.fpi} with (continuous) positive concave mappings, and we recall that these mappings are also standard \cite[Proposition~1]{renato2016}. In particular, the next proposition proves that the spectral radii of asymptotic mappings can be used to obtain a lower bound for the estimation error of the sequence generated by \refeq{eq.fpi} -- see the inequality in \refeq{eq.conc_bound}. 

\begin{proposition}
	\label{prop.conc_bound}
	Assume that $T:\real_+^N\to\real_{++}^N$ is continuous and concave with $\emptyset\neq\mathrm{Fix}(T)=:\{\signal{x}^\star\}$, and denote by $\signal{v}\in\real_{+}^N\backslash\{\signal{0}\}$ any vector satisfying $T_\infty(\signal{v})=\rho(T_\infty)\signal{v}$ (a vector with this property always exists \cite{nussbaum1986convexity}). To simplify notation,  define $\rho:=\rho(T_\infty)<1$, where the inequality follows from Fact~\ref{fact.existence}. Then each of the following holds:
	\begin{itemize}
		\item[(i)] $(\forall n\in\Natural)(\forall\epsilon>0)~T^n(\signal{x}^\star+\epsilon\signal{v})\ge\signal{x}^\star+\rho^n\epsilon\signal{v}$
		\item[(ii)] $(\forall n\in\Natural)(\forall\epsilon>0)~\signal{x}^\star\ge\epsilon\signal{v}\Rightarrow	
		T^n(\signal{x}^\star-\epsilon\signal{v})\le\signal{x}^\star-\rho^n\epsilon\signal{v}$
		\item[(iii)] If $\signal{x}_1\in\real_+^N$ is such that $\signal{x}_1\le \signal{x}^\star-\epsilon\signal{v}$ or $\signal{x}_1\ge \signal{x}^\star+\epsilon\signal{v}\ge\signal{0}$ for some $\epsilon>0$, then 
		\begin{align}		
		\label{eq.conc_bound}
\rho^n\epsilon\|\signal{v}\|\le\|T^n(\signal{x}_1)-\signal{x}^\star\|\to0
\end{align}
 for any monotone norm $\|\cdot\|$ and every $n\in\Natural$.
		
	\end{itemize}
\end{proposition}

\begin{proof}
	(i) We prove the result by induction on $n$. By Lemma~\ref{lemma.inf_conc}, we know that
	\begin{align}
	\label{eq.conc_ineq}
	(\forall \signal{d}\in\real_+^N)(\forall \signal{x}\in\real_+^N) T_\infty(\signal{d})\le T(\signal{x}+\signal{d})-T(\signal{x}).
	\end{align}
	In particular, for $\signal{d}=\epsilon\signal{v}$ and $\signal{x}=\signal{x}^\star$ with $\epsilon>0$ arbitrary, we have 
	$\rho\epsilon\signal{v}=T_\infty(\epsilon\signal{v})\le T(\signal{x}^\star+\epsilon\signal{v})-\signal{x}^\star$, which shows that the desired inequality is valid for $n=1$. Now assume that $T^n(\signal{x}^\star+\epsilon\signal{v})\ge\signal{x}^\star+\rho^n\epsilon\signal{v}$ is valid for an arbitrary $n\in\Natural$. As a consequence of the monotonicity of $T$, we have
	\begin{align}
	\label{eq.monotonicity}
	T^{n+1}(\signal{x}^\star+\epsilon\signal{v})\ge T(\signal{x}^\star+\rho^n\epsilon\signal{v}).
	\end{align}
	Now substitute $\signal{d}=\rho^n\epsilon\signal{v}$ and $\signal{x}=\signal{x}^\star$ into \refeq{eq.conc_ineq} and use the positive homogeneity property of $T_\infty$ to verify that $\rho^{n+1}\epsilon\signal{v}=T_\infty(\rho^n\epsilon\signal{v})\le T(\signal{x}^\star+\rho^n\epsilon\signal{v})-\signal{x}^\star$, and thus $\rho^{n+1}\epsilon\signal{v}+\signal{x}^\star\le T(\signal{x}^\star+\rho^n\epsilon\signal{v})$. Combining this inequality with that in \refeq{eq.monotonicity}, we obtain the desired result $T^{n+1}(\signal{x}^\star+\epsilon\signal{v})\ge \signal{x}^\star+\rho^{n+1}\epsilon\signal{v}$, and the proof is complete.
	
	(ii) The proof is similar to that in part (i), so it is omitted for brevity. 
	(iii) If $\signal{x}_1 \ge \signal{x}^\star+\epsilon\signal{v}\ge \signal{x}^\star$ for some $\epsilon>0$, then $(\forall n\in\Natural)~T^n(\signal{x}_1) \ge T^n(\signal{x}^\star+\epsilon\signal{v})\ge T^n(\signal{x}^\star)=\signal{x}^\star$ by monotonicity of $T$. Therefore, $(\forall n\in\Natural)~T^n(\signal{x}_1) - \signal{x}^\star \ge T^n(\signal{x}^\star+\epsilon\signal{v})-\signal{x}^\star\ge \rho^n\epsilon\signal{v}$, where the last inequality follows from part (i). Monotonicity of the norm $\|\cdot\|$ now shows that  $(\forall n\in\Natural)\|T^n(\signal{x}_1) - \signal{x}^\star\| \ge \rho^n\epsilon \|\signal{v}\|$. In addition, positive concave mappings are standard interference mappings \cite[Proposition~1]{renato2016}, so  $\|T^n(\signal{x}_1)-\signal{x}^\star\|\to 0$ by \cite[Theorem~2]{yates95}, and the proof for $\signal{x}_1 \ge \signal{x}^\star+\epsilon\signal{v}\ge \signal{x}^\star$ is complete. We skip the proof for $\signal{x}_1 \le \signal{x}^\star-\epsilon\signal{v}\le \signal{x}^\star$ because it is similar.
\end{proof}

\section{Numerical example}

To illustrate the results obtained in the previous section in a concrete application, we study the convergence speed of a well-known algorithm for load estimation in wireless networks \cite{Majewski2010,siomina12,feh2013,siomina2014,renato14SPM,renato2016,renato17globalsip}. The algorithm is simply the iteration in \refeq{eq.fpi} with the concave mapping given by $T:\real_+^N\to\real^N_{++}:\signal{x}\mapsto [t_1(\signal{x}),\cdots,t_N(\signal{x})]$, where, for all $i\in\{1,\ldots,N\}=:\mathcal{M}$ and all $\signal{x}\in\real_{+}^N$,

\begin{align}
\label{eq.load}
	t_i(\signal{x}):=\sum\limits_{j\in\mathcal{N}_i}\dfrac{d_j}{KB\log_2\left(1+\dfrac{p_i g_{i,j}}{\sum\limits_{k\in\setm\backslash\{i\}} x_k p_k g_{k,j}+\sigma^2}\right)},
\end{align}
$\mathcal{M}$ is the set of base stations, $\mathcal{N}_i\neq\emptyset$ is the set of users connected to base station $i$, $d_j\in\real_{++}$ is the traffic (in bits/s) requested by the $j$th user, $K\in\Natural$ is number of resource blocks in the system, $B\in\real_{++}$ is the bandwidth per resource block, $p_i>0$ is the transmit power per resource block of base station $i$, $g_{i,j}>0$ is the pathloss between base station $i$ and user $j$, and $\sigma^2>0$ is the noise power per resource block. The $i$th component $x_i^\star$ of the fixed point $\signal{x}^\star\in\mathrm{Fix}(T)$, if it exists, shows the fraction of resource blocks that base station $i$ requires to satisfy the traffic demand of its users. Although we cannot have $x^\star_i>1$ in real network deployments, knowledge of these values is useful to rank base stations according to the unserved traffic demand \cite{siomina12}. See \cite{Majewski2010,siomina12,renato14SPM,ho2014data,renato2016,renato17globalsip,feh2013} for additional details on the load estimation problem. 

The asymptotic mapping associated with $T$ is given by \cite{renato17globalsip} $T_\infty:\real_{+}^N\to\real_{+}^N:\signal{x}\mapsto \mathrm{diag}(\signal{p})^{-1}\signal{M}\mathrm{diag}(\signal{p})\signal{x}$,
where $\mathrm{diag}(\signal{p})\in\real_+^{N\times N}$ is a diagonal matrix with diagonal elements obtained from the components of the power vector $\signal{p}:=[p_1,\ldots,p_N]$, and the component $[\signal{M}]_{i,k}$ of the $i$th row and $k$th column of the matrix $\signal{M}\in\real^{N\times N}_+$ is given by $[\signal{M}]_{i,k} = 0$ if $i= k$ or $[\signal{M}]_{i,k}= \sum\limits_{j \in \mathcal{N}_i} {\mathrm{ln}(2)  d_j g_{k,j}}/({KB g_{i,j}})$ otherwise. The asymptotic mapping $T_\infty$ is linear, so $\rho(T_\infty)$ is simply the spectral radius of the matrix $\signal{M}$. 

With the results in Proposition~\ref{prop.conc_bound}, we can prove that the convergence rate of the recursion in \refeq{eq.fpi} is expected to decrease with increasing traffic. To this end, let $T^\prime(\signal{x})=\beta T(\signal{x})$ for every $\signal{x}\in\real_+^N$, where $\beta\in\real_{++}$ is a design parameter. This new mapping $T^\prime$ can be obtained, for example, by scaling uniformly the traffic demand of every user by a factor $\beta$, and we assume that $\rho(T_\infty)<1$. We can verify that $\rho(T^\prime_\infty) = \beta\rho(T_\infty)>0$. As a result, in light of the bound in \refeq{eq.conc_bound}, by increasing $\beta$, the algorithm is expected to become increasingly slow as $\rho(T^\prime_\infty) = \beta\rho(T_\infty)<1$ approaches the value one.  (The algorithm diverges if $\beta\rho(T_\infty)\ge 1$.)

Fig.~\ref{fig.conc_mapping} illustrates the above points. It shows the estimation error $\|\signal{x}_n-\signal{x}^\star\|_2$ of the fixed point iteration (FPI) in \refeq{eq.fpi} with $\signal{x}_1=\signal{0}$ and the bound in Proposition~\ref{prop.conc_bound}(iii) for the mappings $T$ and $T^\prime$ described above. The parameter $\beta>1$ for $T^\prime$ was chosen to satisfy $\rho(T^\prime_\infty)=0.99$. For the construction of $T$, we use a scenario similar to that in \cite[Sect.~V-A]{renato2016}. Briefly, we obtained snapshots of a network with 1,500 users requesting a traffic of 300 kbps each, and we picked one snapshot with $\rho(T_\infty)<0.99$, in which case we also have $\mathrm{Fix}(T)\neq\emptyset$ as an implication of Fact~\ref{fact.existence}.  Other parameters of the simulation were the same as those in \cite[Table~I]{renato2016}. To compute the bound in \refeq{eq.conc_bound} for $T$, we set the vector $\signal{v}$  to the right eigenvector of $\mathrm{diag}(\signal{p})^{-1}\signal{M}\mathrm{diag}(\signal{p})$ (obtained by using \refeq{eq.krause_iter}) associated with the eigenvalue $\rho(T_\infty)$. In turn, the scalar $\epsilon$ in \refeq{eq.conc_bound} was set to the largest positive real such that $\signal{x}^\star-\epsilon\signal{v}\in\real_{+}^N$. The bound for $T^\prime$ was constructed in a similar way. As expected, the numerical results in Fig.~\ref{fig.conc_mapping} are consistent with the theoretical findings.

\begin{figure}
	\begin{center}
		\includegraphics[width=\columnwidth]{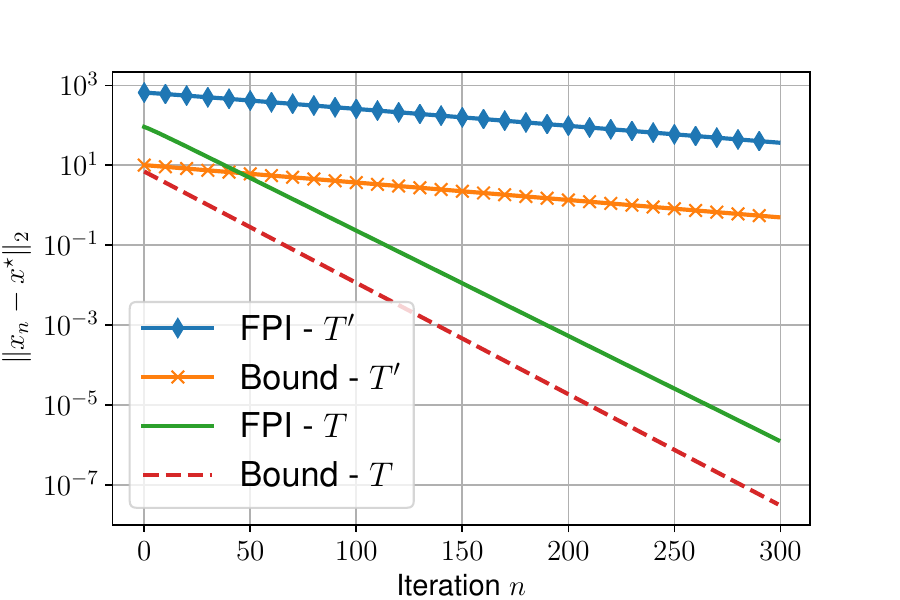}
		\caption{Estimation error as a function of the number of iterations.}
		\label{fig.conc_mapping}
	\end{center}
\end{figure}

\section{Conclusions and Final Remarks}
 
We have shown that knowledge of the spectral radius of asymptotic mappings is useful to relate standard and contractive interference mappings, and with this knowledge we also obtain information about the convergence speed of widely used instances of the recursion in \refeq{eq.fpi}. One advantage of the analysis shown here over existing results in the literature is that we do not assume the mapping $T$ in \refeq{eq.fpi} to be constructed by combining a finite number of affine functions. Furthermore, unlike previous results, in the proposed approaches the parameters used to obtain bounds for the convergence speed are easy to estimate. The  bounds derived here show, for example, that the converge speed of a well-known iterative algorithm for load estimation in wireless networks is expected to decrease with increasing traffic.

\bibliographystyle{IEEEbib}
\bibliography{IEEEabrv,references}

\end{document}

%% file: macros.tex
\usepackage{ifpdf}
\usepackage{manfnt}
\usepackage{tikz}
\usetikzlibrary{arrows,automata,shapes.geometric, positioning}
\usepackage[mathscr]{eucal}
\usepackage{graphicx}
\graphicspath{ {./fig/} }
\usepackage{amsthm}
\usepackage[makeroom]{cancel}
\usepackage{cite}
\usepackage{algorithmicx}
\usepackage{caption}
\usepackage[ruled,vlined,commentsnumbered]{algorithm2e}
\usepackage{empheq}
\usepackage{amssymb}
\usepackage{subfigure}
\usepackage{color}
\usepackage{colortbl}
\definecolor{colorref}{rgb}{0.4648,0,0} 
\definecolor{colorcite}{rgb}{0,0.2902,0.1765}
\usepackage[colorlinks,linkcolor=colorref,urlcolor=blue,citecolor=colorcite]{hyperref}
\usepackage{multirow} 
\usepackage{appendix}
\usepackage{epstopdf}
\usepackage{enumerate}
\usepackage{bm}

\usepackage{amsthm}
\usepackage{amsmath}

 \usepackage[printonlyused,nohyperlinks]{acronym} 
\makeatletter
\def\th@newremark{\th@remark\thm@headfont{\bfseries}}
\makeatletter
\newtheorem{proposition}{Proposition}

\newtheorem{lemma}{Lemma}

\theoremstyle{newremark}

\theoremstyle{newremark}

\theoremstyle{newremark}

\theoremstyle{definition}

\newcommand{\set}[1]{{\mathcal{#1}}}

\newcommand{\setm}{{\set{M}}}

\usepackage{enumitem}

\newcommand{\signal}[1]{{\boldsymbol{#1}}}

\newcommand{\real}{{\mathbb R}}

\newtheorem{definition}{Definition}

\newtheorem{fact}{Fact}

\newcommand{\Natural}{{\mathbb N}}
\renewcommand{\refeq}[1]{(\ref{#1})}
